%% file: main.tex
 \documentclass[sigconf]{acmart}

\copyrightyear{}
\acmYear{}
\setcopyright{none}
\acmConference[arXiv]{Manuscript under review, arXiv preprint}{September}{2024}
\acmBooktitle{arXiv}\acmDOI{}
\acmISBN{}

\usepackage[utf8]{inputenc} 
\usepackage[T1]{fontenc}    
\usepackage{array,url,booktabs,microtype,cite,multirow,hyperref,subcaption,graphicx,mathtools,epstopdf,xspace,amsthm,amsfonts,nicefrac,xcolor,makecell,adjustbox}
\usepackage{wrapfig,lipsum,booktabs}
\usepackage{purudefs}
\usepackage{tabularx}
\usepackage{subcaption}

\usepackage{breqn}
\usepackage{algorithm}
\usepackage{algpseudocode}
\usepackage{algorithmicx}
\usepackage{enumitem} 
\usepackage[skip=0.5\baselineskip]{caption}

\input{definitions}
\input{math_commands}

\begin{document}

\input{sections/frontmatter}

\maketitle

\input{sections/introduction}

\input{sections/related_works}

\input{sections/method}

\input{sections/implementation}
\input{sections/exps}
\input{sections/conclusion}
\input{sections/ethics}
\balance

\bibliographystyle{ACM-Reference-Format}
\balance
\bibliography{references}

\clearpage
\appendix
\input{sections/suppl}
\clearpage

\end{document}

%% file: definitions.tex
\newtheorem*{assumption*}{\assumptionnumber}

\usepackage{array}
\newcolumntype{P}[1]{>{\centering\arraybackslash}p{#1}}

\renewcommand{\it}[1]{\textit{#1}}

\newcommand{\alg}{{CROSS-JEM}\xspace}
\newcommand{\sce}{{monoBERT}\xspace}
\newcommand{\algtitle}{CROSS-JEM: Accurate and Efficient Cross-encoders for Short-text Ranking Tasks}
\newcommand{\rpllong}{{Ranking Probability Loss}\xspace}
\newcommand{\rplshort}{{RPL}\xspace}

\definecolor{arsenic}{rgb}{0.23, 0.27, 0.29}

\definecolor{silver}{rgb}{0.75, 0.75, 0.75}

\definecolor{ForestGreen}{HTML}{228B22}
\definecolor{Maroon}{HTML}{800000}

\newcommand{\good}[1]{{\textcolor{ForestGreen}{#1}}}
\newcommand{\bad}[1]{{\textcolor{Maroon}{#1}}}
\newcommand{\meh}[1]{{\textcolor{arsenic}{#1}}}

%% file: math_commands.tex
\usepackage{amsmath,amsfonts,bm,mathrsfs}

\newtheorem*{theorem*}{Theorem}
\newtheorem*{lemma*}{Lemma}
\newtheorem*{corollary*}{Corollary}

\def\eqref#1{equation~(\ref{#1})}
\def\Eqref#1{Equation~(\ref{#1})}

\def\1{\bm{1}}

\def\rve{{\mathbf{e}}}

\def\rvk{{\mathbf{k}}}

\def\rvq{{\mathbf{q}}}

\def\rvt{{\mathbf{t}}}

\def\vtheta{{\bm{\theta}}}

\def\ve{{\bm{e}}}
\def\vf{{\bm{f}}}

\def\vk{{\bm{k}}}

\def\vp{{\bm{p}}}
\def\vq{{\bm{q}}}

\def\vs{{\bm{s}}}

\def\vu{{\bm{u}}}

\def\vw{{\bm{w}}}
\def\vx{{\bm{x}}}
\def\vy{{\bm{y}}}

\DeclareMathAlphabet{\mathsfit}{\encodingdefault}{\sfdefault}{m}{sl}
\SetMathAlphabet{\mathsfit}{bold}{\encodingdefault}{\sfdefault}{bx}{n}

\def\gE{{\mathcal{E}}}

\def\gT{{\mathcal{T}}}

\def\sI{{\mathbb{I}}}

\def\sK{{\mathbb{K}}}
\def\sL{{\mathbb{L}}}

\def\sP{{\mathbb{P}}}
\def\sQ{{\mathbb{Q}}}
\def\sR{{\mathbb{R}}}
\def\sS{{\mathbb{S}}}
\def\sT{{\mathbb{T}}}



%% file: sections/frontmatter.tex
\title{\algtitle}

\author{Bhawna Paliwal} \authornote{Equal contribution} \authornote{Corresponding authors}
\email{bhawna@microsoft.com}
\affiliation{%
  \institution{Microsoft Research}
  \country{Bengaluru, India}
}

\author{Deepak Saini} \authornotemark[1] \authornotemark[2]
\email{desaini@microsoft.com}
\affiliation{%
  \institution{Microsoft}
  \country{Redmond, USA}
}

\author{Mudit Dhawan} \authornote{Work done while at Microsoft Research}
\email{muditd2000@gmail.com}
\affiliation{%
  \institution{Carnegie Mellon University}
  \country{Pittsburgh, USA}
}

\author{Siddarth Asokan} \authornotemark[2]
\email{Siddarth.Asokan@microsoft.com}
\affiliation{%
  \institution{Microsoft Research}
  \country{Bengaluru, India}
}

\author{Nagarajan Natarajan}
\email{Nagarajan.Natarajan@microsoft.com}
\affiliation{%
  \institution{Microsoft Research}
  \country{Bengaluru, India}
}

\author{Surbhi Aggarwal}
\email{Surbhi.Aggarwal@microsoft.com}
\affiliation{%
  \institution{Microsoft}
  \country{Bengaluru, India}
}

\author{Pankaj Malhotra}
\email{pamalhotra@microsoft.com}
\affiliation{%
  \institution{Microsoft}
  \country{Bengaluru, India}
}

\author{Jian Jiao}
\email{Jian.Jiao@microsoft.com}
\affiliation{%
  \institution{Microsoft}
  \country{Redmond, USA}
}

\author{Manik Varma}
\email{manik@microsoft.com}
\affiliation{%
  \institution{Microsoft Research}
  \country{Bengaluru, India}
}

\renewcommand{\shortauthors}{Paliwal and Saini, et al.}

\keywords{sponsored search, information retrieval, keyword scoring, keyword ranking, cross encoders, efficiency, production systems, online systems, industrial applications, large-scale learning}

\begin{abstract}

Ranking a set of items based on their relevance to a given query is a core problem in search and recommendation. Transformer-based ranking models are the state-of-the-art approaches for such tasks, but they score each query-item independently, ignoring the joint context of other relevant items. This leads to sub-optimal ranking accuracy and high computational costs. In response, we propose Cross-encoders with Joint Efficient Modeling (\alg), a novel ranking approach that enables transformer-based models to jointly score multiple items for a query, maximizing parameter utilization. \alg leverages (a) redundancies and token overlaps to jointly score multiple items, that are typically short-text phrases arising in search and recommendations, and (b) a novel training objective that models ranking probabilities. \alg achieves state-of-the-art accuracy and over 4x lower ranking latency over standard cross-encoders. Our contributions are threefold: (i) we highlight the gap between the ranking application's need for scoring thousands of items per query and the limited capabilities of current cross-encoders; (ii) we introduce \alg for joint efficient scoring of multiple items per query; and (iii) we demonstrate state-of-the-art accuracy on standard public datasets and a proprietary dataset. \alg opens up new directions for designing tailored early-attention-based ranking models that incorporate strict production constraints such as item multiplicity and latency. 
\end{abstract}

%% file: sections/introduction.tex
\section{Introduction}
\label{sec:intro}

We consider the problem of ranking that arises in search and recommendation pipelines, wherein the goal is to rank a set of items based on their relevance to a given query. Our work is in the context of two-stage \textit{retrieve-then-rank} pipelines in modern search engines~\citep{liu2017cascade,nogueira2020passage,zhao2022densesurvey,lin2021pretrained1,zhou2022towards,fan2022pre} as depicted in Figure \ref{fig:multi-stage}. Given a \textit{query}, i.e., a search phrase such as \textit{``patagonia japan''}, the retrieval stage pares the \textit{items}, i.e., keywords such as \textit{``clothing store japan''}, \textit{``patagonia homes shimoda'' }that are bid against for displaying ads, from billions to a few hundreds~\citep{guo2022semantic,matveeva2006high} of potential interest. These items are subsequently provided as the input to the ranking stage. In this work, we focus on the ranking model, given a black-box retriever. We consider short-text items (as in the above example), which appear in a myriad of recommendation systems applications such as product recommendation, query to advertiser bid phrase recommendation, and Wikipedia category tagging~\citep{Dahiya21,Zhou19,wikipedia06}. In designing the ranking model, two key axes are the model architecture, and the choice of the loss function, while the key performance metrics for such systems are accuracy and inference latency. \par
\noindent \textbf{Ranking architectures and limitations:} Along the architecture vertical, \textit{encoders} that employ stacked attention layers to encode a query-item pair, followed by a \textit{classifier} to predict ranking scores are widely adopted for ranking~\citep{nogueira2020passage,nogueira2019multi,zhou2022towards}.
Recently, sequence-to-sequence models with encoder-decoder and decoder-only architectures have also been proposed for ranking. These models provide a ranking score based on particular vocabulary-token logits~\citep{nogueira-etal-2020-document,rankt5,zhang-etal-2024-two}. 
However, these approaches model ranking as a \textit{pointwise} task, providing ranking scores for a given query and item pair.
But ranking is inherently a list-based task that requires scoring query-item pairs relative to one another, and not in isolation. In particular, the other items in the list to be ranked for a given query provide crucial context for scoring. Pointwise models neglect the list context, produce independent scores that may not reflect the optimal ranking order, and are difficult to calibrate across items for sorting to provide final rankings~\citep{qin-etal-2024-large}.

Table~\ref{table_IntroCompares} illustrates this by juxtaposing the top-5 ranked items obtained using our proposed approach (listwise modeling) and a baseline pointwise ranking model ~\citep{nogueira2020passage}. We observe that relatively more \textit{generic} items such as ``\textit{mexican cuisine}'', although relevant to the query \textit{``different foods of
oaxaca mexico''}, are ranked higher in baseline predictions, owing to (a) their frequency in training data, (b) token-level matching and other biases which are difficult to mitigate in pointwise modeling. On the other hand, our proposed listwise approach evaluates all the items to be ranked holistically, and subsequently ranks more specific (not just relevant) items higher.

Furthermore, pointwise transformer based models~\citep{nogueira2020passage,rankt5} are computationally expensive and impractical for real-time ranking systems that need to handle large-scale traffic requiring low latency and high throughput. 
Therefore, many industrial systems resort to using simpler sparse neural networks~\citep{MicrosoftResearch2023} or late-interaction models~\citep{lu2020twinbert,khattab2020colbert} for online ranking, sacrificing accuracy for latency.

\begin{figure*}[!htb]
\centering
\includegraphics[width=1.0\linewidth]{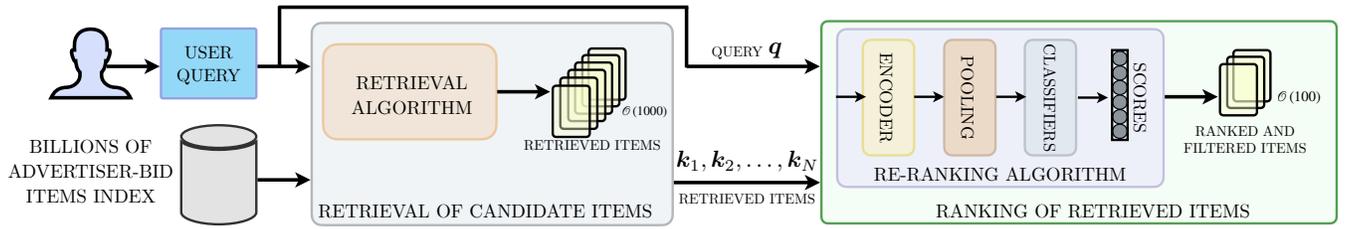}
\caption{The two-stage large-scale search and recommendation pipeline comprises: $(i)$ candidate selection from billions of items; $(ii)$ re-ranking of retrieved items. \alg works at stage $(ii)$, to output the score for all \(N\) item in a single pass.}
\label{fig:multi-stage}
\end{figure*}

\setlength{\tabcolsep}{4pt}
\begin{table*}[!htb]
    \caption{A comparison of \alg{'s} listwise modeling and a pointwise ranking model ~\citep{nogueira2020passage}: Relatively more \textit{generic} (but relevant) items are ranked higher in baseline predictions, owing to biases such as high frequency in the training data. Our listwise model evaluates all the shortlisted items in a single forward pass, and ranks more specific (and relevant) items higher.}
    \label{table_IntroCompares}
      \centering
    \fontsize{6.5}{10}\selectfont
    \begin{tabular}{P{.095\linewidth}|P{.41\linewidth}P{.445\linewidth}}
      \toprule
      \textbf{Query} & \textbf{Top-5 ranked items in the Proposed Approach (\alg)}  & \textbf{Top-5 ranked items in the baseline Cross Encoder}   \\
       \midrule 
       different foods of oaxaca mexico & \good{`the foods of oaxaca'}, \good{`oaxacan cuisine'}, \good{`exploring oaxacan food'}, \good{`authentic recipes from oaxaca, mexico'}, \meh{`6 things you'll love about oaxaca'} & \meh{`culinary tales: the kinds of food mexicans eat'}, \meh{`mexican christmas foods'}, \meh{`mexican cuisine'}, \meh{`culture: food and eating customes in mexico'}, \meh{`popular food in mexico'},  \\
       \midrule 
       what is the bovine growth hormone & \good{`recombinant bovine growth hormone'}, \good{`what is rbgh?'}, \good{`rbgh'}, \meh{`bovine growth hormone and milk : what you need to know'}, \good{`what is rbst?'} & \meh{`growth hormone'}, \bad{`human growth hormone'}, \bad{`alternative names for growth hormone'}, \bad{`human growth hormone and insulin are friends'}, \bad{`growth hormone (somatotropin)'}\\
      \bottomrule
  \end{tabular}
\end{table*}
Another line of research has been along the loss function, wherein training with listwise loss functions~\citep{reranker-cmu,rankt5,cao2007learning} have shown accuracy gains. These listwise losses optimize the models for a list of items to be ranked, rather than for individual query-item pairs, and can enhance the ranking performance without increasing the model size or complexity. However, these model architectures still operate at the query-item (pointwise) level, and produce independent scores for each item in the list, without explicitly modeling the inter-item dependencies or the query context. 
Recent works also incorporate listwise modeling via pre-trained LLMs for ranking~\citep{sun2023chatgpt,pradeep2023rankzephyr,qin2024large,zhang-etal-2024-two}. 
These approaches typically work with pre-trained large-scale models, focus on specifying all ranking items in the prompt and employing prompt engineering methods to improve accuracy and minimize LLM inference calls. 
Due to the huge parameter count (running into a few billions), these models cannot be deployed in large-scale online ranking systems. \textbf{In this paper, we aim to bridge this gap by proposing a ranking model that works at the list level, explicitly models inter-item interactions, and achieves a superior latency-accuracy tradeoff, making it deployable in real-time scenarios.} \par

\noindent \textbf{The Proposed Approach}: We propose an end-to-end joint ranking approach that models the listwise ranking of the query and given items, capturing both the query-item and the item-item interactions and satisfying the strict latency requirements of industry-scale recommendation systems. Our approach, entitled \textbf{\alg} (\textbf{CROSS}-Encoder with \textbf{J}oint \textbf{E}fficient \textbf{M}odeling), leverages the list structure of the input in both the encoder and the classifier components, as well as in the training objective. Specifically, our encoder is a transformer-based architecture that processes a query and a list of items to rank in a single pass, generating \textit{`list-aware'} context vectors for each input token. To achieve this, each item token attends to all other items in the list, as well as to the query, allowing the encoder to capture the joint relevance of all items. This is followed by a classifier which is jointly trained with the encoder. While any standard loss function (e.g., binary cross entropy) could be used to train a \alg model, in this work, we propose a novel variant of the listwise loss functions~\citep{cao2007learning,qin2008query-level}, \rpllong{} (\rplshort); which can be interpreted as a divergence between the predicted probabilities and the estimated ranking probabilities (as a function of model logits and ground truth rankings) of items to be ranked. The proposed listwise loss works much better in conjunction with our joint modeling than standard pointwise or listwise losses~\citep{cao2007learning,qin2008query-level,NIPS2009_2f55707d}. To the best of our knowledge, we are the first to propose a joint ranking approach that can effectively model listwise ranking in both the model architecture and the training objective with real-time latency constraints.

\alg is able to support low latency applications (few ms) by significantly reducing the computational costs of cross-attention across query and the list of ranking items.
This is achieved by exploiting the presence of token duplicates in the retrieved set of items for a given query. 
\alg exploits this redundancy to sidestep processing long sequences, thereby keeping the inference latency small. Further, jointly obtaining the ranking scores for a list of items avoids multiple calls to the expensive encoder (unlike pointwise approaches) making \alg significantly faster.

In summary, our key contribution is the introduction of a novel \textbf{joint ranking approach \alg}, that scores multiple items per query in a single pass, exploiting token interaction across items for better accuracy and token redundancies for better efficiency.
We demonstrate the effectiveness of \alg on two public benchmark datasets for short-text re-ranking, wherein it outperforms the best-performing baselines by at-least 3\% in terms of MRR. When applied to large-scale search-based recommendation, \alg demonstrated a 13\% higher accuracy than state-of-the-art models, while being over 6$\times$ faster than standard cross-encoders~\citep{nogueira2020passage}. 
We also deploy \alg for real-time ranking on live traffic, where it reduces the quick-back-rate by 1.8\%, indicating improved relevance of ads to users. 
Our work presents a general, scalable framework for joint ranking of multiple items across various domains and short-text tasks, accounting for ranking under task-specific constraints such as item multiplicity and latency.

%% file: sections/related_works.tex
\section{Background and Related Work}\label{sec:related_work}

\textbf{Ranking Architecture}: State-of-the-art transformer based models with stacked attention layers are typically encoder based, such as monoBERT and Birch~\citep{nogueira2020passage,akkalyoncu-yilmaz-etal-2019-cross}. These models encode query-passage pairs with a bi-directional attention encoder and use a classifier to obtain a ranking score. Alternatively, sequence-to-sequence models, such as monoT5~\citep{nogueira-etal-2020-document} leverage the pre-training knowledge of generative models for ranking. These models generate a ranking score from a specific vocabulary token in the decoder output. Another line of work explores decoder-only models, such as llama- and GPT-based models and rely on their extensive pre-training knowledge and parameter count for ranking~\citep{zhang-etal-2024-two,fine-tune-llama}. Despite their advantages, decoder-only models fine-tuned on large ranking corpora~\citep{zhang-etal-2024-two} do not outperform the fine-tuned encoder-based and sequence-to-sequence models on short-text ranking tasks (cf. Section~\ref{sec:exps}). \alg is the first joint-ranking approach that effectively incorporates listwise ranking into the model architecture and training objective while maintaining latency constraints.\\
\noindent \textbf{Joint (Listwise) Ranking}: 
Ranking inherently involves comparing a list of items; thus, recent works employ listwise loss functions in encoder or encoder-decoder models to enhance ranking accuracy. 
~\citet{gao2021rethink} devised a listwise multi-class cross-entropy loss to optimize ranking probabilities in encoder models. 
~\citet{rankt5} introduced a ranking-specific listwise cross-entropy loss to improve performance of Seq2Seq models for ranking. They also demonstrated that an expanded form of cross-entropy loss (poly-1) achieved superior performance across various ranking metrics.
Although listwise losses improved ranking accuracy, the model architectures of these rankers remain pointwise and do not fully capture the ranking task. 
\alg ranks multiple items per query in one pass, capturing item-item interactions in the ranking list to improve accuracy. It is optimized with \rpllong (\rplshort), a novel variant of the listwise ranking loss (as considered in ListNet~\citep{cao2007learning}) aligned with the \alg architecture, and estimates the ranking probabilities using target relevance labels and predicted model logits. \\
\textbf{LLMs for Ranking}: LLMs have emerged as a powerful tool for ranking tasks; they can leverage pre-trained knowledge, large parameter count, and effective prompting to achieve superior performance. Existing works have adopted two main approaches to exploit LLMs for ranking: (a) Using LLMs directly as re-rankers by designing novel prompting schemes~\citep{ma2023zero, sun2023chatgpt, pradeep2023rankzephyr} and sorting strategies~\citep{qin2024large, zhuang2024setwise}; the high-level idea is to encode the ranking items and the query into a single input and generate a ranked list as output, using a sliding window technique to handle long inputs; (b) Using LLMs for learning more accurate smaller ranking models \citep{sun2023chatgpt, pradeep2023rankvicuna,zhang-etal-2024-two, ma2024fine} via standard distillation or via augmenting the training data with synthetic samples generated by the LLM. These works have demonstrated that LLMs can outperform small-scale supervised methods~\citep{nogueira2020passage, nogueira-etal-2020-document,qin2024large} on various ranking benchmarks. Deploying LLMs at scale for real-time serving scenarios is still challenging due to their high computational costs and memory requirements. But using LLMs to enhance smaller ranking models, as in (b), is more practical and scalable --- this can be applied to our proposed \alg as well.

\noindent \textbf{Efficiency and Scaling}: 
Another closely related area is the focus on improving the efficiency and scaling of transformer based rankers (and retrievers) using light-weight architectures. 
Approaches employing early-attention (such as monoBERT~\citep{nogueira2020passage}) have shown high ranking accuracies, but cannot support the low-latency requirements of online raking applications in industry-scale recommendation systems. 
This can be attributed to the requirement of making multiple calls to the expensive transformer-based encoder to rank a list of items per query, which is infeasible in a few milliseconds.
Therefore, online production systems use a variation of sparse neural networks, namely  MEB~\citep{MicrosoftResearch2023}, for ranking thousands of items in real-time.
Late-interaction models such as ColBERT~\citep{khattab2020colbert}, Baleen~\citep{khattab2021baleen} and TwinBERT~\citep{lu2020twinbert} are also used for online ranking due to their computational efficiency. 
These models reduce the computational costs by applying a late-interaction layer over query-item embeddings. This comes at the price of accuracy, owing to the lack of interactions between query and item tokens. They also incur high storage overheads in online settings, as they need to store and retrieve token-level embeddings. 
Another line of work that focuses on efficient retrieval and ranking is based on dual-encoder architectures, such as ANCE~\citep{Xiong20}, DPR~\citep{Karpukhin20}, and INSTRUCTOR~\citep{INSTRUCTOR}. 
These methods use contrastive-style training to learn a query and item encoder and metrics such as cosine similarity to rank query-item pairs. 
They can be made highly efficient via a nearest neighbor search~\citep{Subramanya19,malkov2018efficient} over pre-computed embeddings, but lose out on accuracy compared to cross-encoder based approaches~\citep{nogueira2020passage,santhanam-etal-2022-colbertv2}.

%% file: sections/method.tex
\section{\underline{CROSS}-encoder with  \underline{J}oint \underline{E}fficient \underline{M}odeling}
\label{sec:method}
We now present \alg, our proposed approach to accurate and low-latency ranking via efficient scoring of multiple items. Before we describe \alg in detail, we develop the notation used in the subsequent sections of this manuscript. 

\noindent \textbf{Notation:} We denote queries by $\vq$ and the corresponding set of $N$ candidate items retrieved for $\vq$ by $\sK_{\vq}  = \{\vk_1, \vk_2, ..., \vk_{N}\}$. We denote the dataset of queries and items used for training by $\sQ_{tr}$ and $\sI_{tr}$ respectively, and that of the test datasets by $\sQ_{te}$ and $\sI_{te}$. We drop the subscripts when the meaning is clear from the context. The ground-truth scores are given by \(\vy_{i} \in \sR^N\), \([\vy_i]_j = y_{ij}\) denoting the score for the item \(\vk_j\) from set \(\sK_{\vq_i}\), associated with the query \(\vq_i\). The queries \(\vq\) and items \(\vk_j\) are tokenized via \(\gT(\cdot)\) to obtain \(d\)-dimensional representations (tokens), given by \(\sT_{\vq} = \left\{\rvq^1,\rvq^2,\ldots,\rvq^{L_{\vq}}\right\}\), and \(\sT_{\vk_{j}} = \left\{\rvk_j^1,\rvk_j^2,\ldots,\rvk_j^{L_{\vk_j}}\right\}\), respectively, where \(\rvq^{\ell},\rvk_j^{\ell} \in \sR^d\). 

\noindent \textbf{Research Problem:} We seek to learn a ranking model that, given a query $\vq$, assigns scores $\sS_{\vq} = \{s_1,s_2,..., s_{N}\}$ for all items in $\sK_{\vq}$, such that the ranking induced by the scores is accurate.

\subsection{The \alg Architecture}
\label{subsec:method_arch}
\alg comprises an encoder to obtain representations of a given query \(\vq\) and all its candidate items $\sK_{\vq}$. The representations are pooled and passed to a classification head, that outputs a score corresponding to each item \(\vk_i\in\sK_{\vq}\) associated with the query. \par
\alg{'s} architecture is primarily inspired by the observation that the candidate items in the set \(\sK_{\vq}\), for a given query \(\vq\), has significant token overlap amongst themselves. A more in-depth exploration of this phenomenon, in the context of efficient-scoring, is provided in Section~\ref{SubSubSec_Efficiency}. Given a query \(\vq\) and its candidate items \(\sK_{\vq}\), the core idea is we form the union of tokens \(\sT_{U_{\vq}}\) of all items in \(\sK_{\vq}\). Subsequently, the representations of query and set \(\sT_{U_{\vq}}\) can be obtained in a \textit{single pass} of the encoder, and scored via a \textit{single pass} over the classifiers. \par
While at first glance, it might appear that using \(\sT_{U_{\vq}}\) (with a potential loss of ordering of the item tokens) could adversely affect performance, we observed in preliminary experimentation that this is not the case when scoring short-text items. In particular, we compared the performance of two cross-encoder models on search engine logs, one with items as-is, and another comprising items with alphabetically sorted tokens. Both the mean average precision (MAP) and accuracy of the latter model was found to be within 1\% of the score obtained when the sequence information is retained. 
Additional discussions are provided in Appendix ~\ref{App:AdsExp}. 
\par
We now describe the \alg encoder and classifier in detail. \par

\noindent \textbf{Encoder}: \alg employs a trainable encoder $\cE_{\vtheta}$, which takes as input a sequence of tokens $\vT = \left[\rvt^1, \rvt^2, \ldots, \rvt^{n}\right]$ (of length \(n\)), and generates as output another sequence of $d$-dimensional {\it contextual embeddings} $\cE_{\vtheta}(\vT) = \vE = \left[\rve_1, \rve_2, \ldots, \rve_{n}\right]$. These embeddings provide context-dependent representations of the input tokens, and can be used for downstream tasks such as classification, generation, and retrieval. Given the tokenization of the query \((\sT_{\vq})\), and that of an item \((\sT_{\vk_j})\), the contextual embeddings of the tokens 
\[ 
[\rvt^{\left[\mathrm{CLS}\right]}, \rvq^1, \rvq^2,\ldots, \rvq^{L_{\vq}}, \rvt^{\left[ \mathrm{SEP} \right]}, \rvk_j^1, \rvk_j^2,\ldots, \rvk_j^{L_{\vk_j}}]
\] in baseline variants is given by 
\[
[\rve^{\left[\mathrm{CLS}\right]}, \rve^{\rvq^1}, \rve^{\rvq^2},\ldots, \rve^{\rvq^{L_{\vq}}}, \rve^{\left[\mathrm{SEP}\right]}, \rve^{\rvk_j^1},\rve^{\rvk_j^2},\ldots, \rve^{\rvk_j^{L_{\vk_j}}} ],
\]
where \(\rve^{\left[\mathrm{CLS}\right]}\) and \(\rve^{\left[\mathrm{SEP}\right]}\) denote the embeddings of the \(\rvt^{\left[\mathrm{CLS}\right]}\) and \(\rvt^{\left[ \mathrm{SEP} \right]}\) tokens, respectively. These contextual embeddings are pooled to obtain a single $d$-dimensional embedding for each pair \((\vq,\vk_j)\) by means of a sum or mean pooling layer, or by taking $\rve^{\left[\mathrm{CLS}\right]}$. This process is computationally expensive due to the need for \(N\) forward passes of the encoder to compute the scores for each item in $\sK_{\vq}$. \par
In \alg, we leverage the short-text nature of the items, item-item interactions, and redundancy of tokens amongst items in $\sK_{\vq}$. This is done by computing the contextual embeddings for all \textit{distinct tokens in the retrieved item set $\sK_{\vq}$} in one pass over the sequence of the query tokens \(\sT_{\vq}\) combined with \textit{item token union set} $\sT_{U_{\vq}} = \left\{\rvt^{\vu^1}, \rvt^{\vu^2},\ldots, \rvt^{\vu^{M}}\right\}$, where $M$ is the total number of tokens in the union set. The contextual embeddings of all input tokens in \alg are computed as 
\begin{align*}
    \mathbf{E}&=\cE_{\vtheta}\left(\left[\rvq^1\!\!, \ldots, \rvq^{L_{\vq}}\!, \rvt^{\left[ \mathrm{SEP} \right]}, \rvt^{\vu^1}\!\!\!\!, \ldots, \rvt^{\vu^{M}}\right]\right) \\
    &= \Big[\rve^{\left[\mathrm{CLS}\right]}, \rve^{\rvq^1}\!\!, \ldots, \rve^{\rvq^{L_{\vq}}}\!, \rve^{\left[\mathrm{SEP}\right]}\!, \rve^{\rvt^{\vu^1}}\!\!,\ldots, \rve^{\rvt^{\vu^{M}}} \Big].
\end{align*}
Since the number of tokens in the item union set is significantly smaller than the sum of tokens of all items in $\sK_{\vq}$ (cf. Section \ref{SubSubSec_Efficiency}), the proposed token-union-based inference enables highly efficient computation of contextual embeddings. 
Figure~\ref{fig:inference} (a) illustrates the difference between the \alg encoder, and standard encoders such as monoBERT. 

\begin{figure*}[t!]
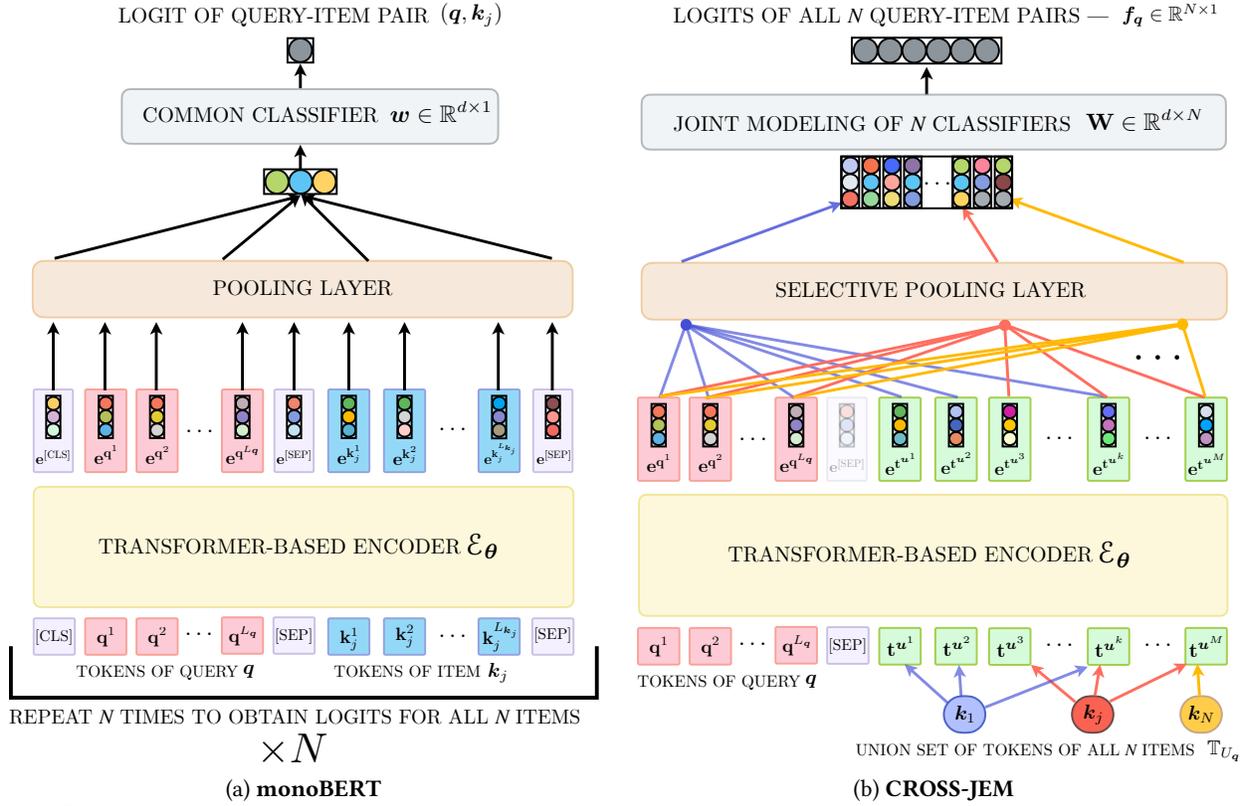

  \begin{center}
    \begin{tabular}[c]{P{.45\linewidth}P{.46\linewidth}}
      \includegraphics[width=1.0\linewidth]{figures/SID_MonoBERT_WSDM.pdf} & 
      \includegraphics[width=1.0\linewidth]{figures/SID_CrossJEM_WSDM.pdf} \\
      (a) \textbf{\sce} & (b) \textbf{\alg}\\[-5pt]
    \end{tabular} 
  \caption[]{ (\includegraphics[height=0.013\textheight]{figures/Rgb.png} Best viewed in color)~ Computing the relevance scores for query $\vq$ and retrieved set of $N$ items $\{\vk_1,\vk_2,\ldots,\vk_{N}\}$ using (a) \sce and (b) \alg: (a) \sce (standard cross-encoders) infers the scores for each $(\vq, \vk_j)$ pair with query and item tokens individually (pointwise). The inference is repeated $N$ times for $N$ items. (b) \alg jointly inputs query tokens and the set of union of tokens from all items to be ranked (listwise). The logits for an item $\vk_j$ can be computed by selecting the contextual embeddings corresponding to its tokens. The selected embeddings are pooled, and input to the linear classifier jointly (\textit{i.e.,} \(\mathbf{W} = [\vw, \vw, \ldots, \vw] \in \mathbb{R}^{d\times N}\)), thus obtaining the $N$ scores in a single encoder and classifier pass.
   }
  \label{fig:inference}  
  \end{center}
  \end{figure*}

\noindent \textbf{Selective Pooling Layer}: Given the contextual embeddings $\mathbf{E}$ for all query tokens $(\sT_{\vq})$ and union over keyword tokens $(\sT_{U_{\vq}})$, \alg employs a selective pooling layer to \textit{jointly model} the per-item relevance score. 
Given \( \vk_j \in \sK_{\vq}\), a pooled representation for pair $(\vq, \vk_j)$ is computed as the mean of the contextual embeddings for all tokens in $\sT_{\vq}$ and those tokens in $\sT_{U_{\vq}}$ which are present in $\sT_{\vk_j}$. The set of selected tokens for the pooling layer is given by
\(
    \sP_{\vq\vk_j} = \sT_{\vq} \bigcup \left\{ \rvt^{\left[SEP\right]}\right\} \bigcup \left\{\sT_{U_{\vq}} \bigcap \sT_{\vk_j}\right\}. 
\)
The selectively pooled representation $\ve_{qk_j} \in \mathbb{R}^d$ is obtained via selective mean pooling:
\begin{equation}
    \ve^{\vq\vk_j} = \frac{1}{|\sP_{\vq\vk_j}|} \sum_{\rvt^j \in \sP_{\vq\vk_j}} \ve^{\rvt^j}.
\end{equation}
Ablations on designing the pooling layer are provided in Section~\ref{sec:exps}.

\noindent \textbf{Classifier}: The final stage in \alg is a $d$-dimensional shared linear classifier $\vw \in \mathbb{R}^d$ which computes the relevance score associated with each pair \((\vq,\vk_j)\). The selectively pooled representations $\ve^{\vq\vk_j}$ obtained for all $\vk_j \in \sK_{\vq}$ are batched together ($\ve^{\vq\vk} \in \mathbb{R}^{N \times d}$) allowing for the computation of all logits \(
    [\vf_{\vq}]_j =  \langle \vw, \ve^{\vq\vk_j} \rangle
\) in a single shot. The scores are defined over these logits (cf. Section~\ref{subsec:rpl}).

%% file: sections/implementation.tex
\section{The \alg Algorithm} \label{sec:implement}

For \textbf{training} \alg, we jointly learn the encoder model and classifier parameters $\{\vtheta, \vw\}$ with target scores obtained from a teacher model. The teacher is a large cross-encoder-based model which is highly accurate but computationally expensive (cf. Section~\ref{sec:exps}). The encoder followed by the selective pooling layer output representations associated with the query and each candidate item, while the linear classifier generates logits for each query-item pair. During \textbf{inference}, \alg predicts the logits for a set of retrieved shortlist items jointly (for a given query), as in the training phase. The \(\mathrm{SoftMax}\) over the logits give the scores, and in turn, the ranking. The inference algorithm is provided in Appendix~\ref{sup:algo}.
\par

The \rpllong (\rplshort) used to train \alg is a novel variant of existing list-based loss functions, and is designed to take advantage of CROSS-JEM architecture, wherein all item scores are available jointly from a single forward pass.

\subsection{The \alg Ranking Probability Loss} \label{subsec:rpl}
The standard loss to train cross-encoders is the binary cross-entropy loss, defined over \((\vq_i,\vk_j)\), given by \(\cL^{BCE}(\vtheta, \vw)\) is
\begin{align}
     -\sum_{i=1}^{|\sQ_{tr}|} \sum_{j=1}^{N} &\quad \Big[ y_{ij} \log \left(\vf_{\vw,\vq_{i},j}\right) + \left( 1 - y_{ij}\right) \log\left( 1 - \vf_{\vw,\vq_{i},j}\right) \Big],
     \label{equation:loss_func}
\end{align}
where \(\vf_{\vw,\vq_{i},j} = \left\langle\vw, \cE_{\vtheta}\left(\vq_i, \vk_j \right) \right\rangle\) is the score of item \(j\), associated with query \(i\), computed by means of an inner product with the classifier \(\vw\).
However, such cross-entropy-based pointwise losses fail to account for the list of items available for ranking. \textit{List-based loss functions}~\citep{cao2007learning}, in contrast, leverage the task-specific ranking information, help learn a scoring function for a list of items to be ranked, rather than for individual query-item pairs. As an example, consider the ListNet~\citep{cao2007learning} loss \(\cL^{LN}(\vtheta, \vw)\), given by:
\begin{align}
    \cL^{LN}(\vtheta, \vw) = - \sum_{i=1}^{|\sQ_{tr}|}\sum_{j=1}^{N} P_{\vy,j} \log \left(P_{\vf,j}\right), \label{eqn:LNloss}
\end{align}
where \(\displaystyle P_{\vx,j} = \frac{\Phi\left([\vx]_j\right)}{\sum_{\ell=1}^{N}\Phi\left([\vx]_{\ell}\right)}\), and \(\vx\) is set either to the targets \(\left[\vy_{\vq_{i}}\right]\), or the output logits \(\left[\vf_{\vq_{i}}\right]\), and \(\Phi\) is a normalizing function, typically the exponential operation, leading to \(P\) being a \(\mathrm{SoftMax}\) function. However, this formulation is still centered around obtaining the pointwise logits, and subsequently computing the top-one probability \(P_{\vx,j}\) using a normalization term that accounts for all pairs. \par
In \alg, we design a novel version of the ListNet loss, one that factors in the availability of all logits \(\left[\vf_{\vq_{i}}\right]\) computed by taking into account the item-item interactions. Given $\vf_{\vq_i}$, the ground-truth scores \(\vy_i\), and a candidate item \(\vk_j\), we define the set $\sL_j = \{k \in \{1, N\}: [\vy_i]_k < [\vy_i]_j \}$, \textit{i.e.,} \(\sL_j\) comprises the indices \(k\) for which the ground truth score at location \(k\) is \textbf{lower} than \([\vy_i]_j\), the ground truth score at location \(j\). We now define the \rplshort as: 
\begin{equation}
    \label{eqn:rlp}
    \cL^{\rplshort} = -\sum_{i=1}^{|\sQ_{tr}|}\sum_{j=1}^{N}\Big(\sum_{k \in \sL_{j}}[\vy_i]_{k} \Big)\log\Big(\mathrm{SoftMax}\Big(\sum_{k \in \sL_{j}} [\vf_{\vq_i}]_{k} \Big)\Big).
\end{equation}
The loss \(\cL^{\rplshort}\) represents a cross-entropy loss over functions of the target \(\vy_i\) and scores computed as a function of the logits \(\vf_{\vq_j}\). The following Lemma sheds light on the relationship between \rplshort and the ranking probability distribution.
\begin{lemma}
    \label{lemma:rpl}
    (\textbf{\rpllong}) 
    Assume without loss of generality that \(\vf_{\vq_i} \in [0,1]^N\). Let \(\vP \in \sR^{N\times N}\) denote a matrix with entries \(p_{jk}\) given by \(p_{jk}~=~\mathrm{Prob}\left(\text{ranking item \(\vk_j\) at location $k$}\right) \triangleq C \sum_{\ell \in \sL_{k}} [\vf_{\vq_i}]_{\ell}\), where \(C\) is a normalizing constant. Then, the \rpllong~maximizes the probability of ranking items \(\sK_{\vq_i}\) in the ordering of the ground-truth \(\vy_i\).
\end{lemma}

 The detailed proof is given in Appendix~\ref{app:proofs}. The sketch of the proof follows by evaluating the matrix \(\vP\) for the predicted logits \(\vf_{\vq_i}\) defined above, and \(\vP^*\) defined over the ground-truth scores \(y_{ij}\). Minimizing the distance between \(\vP\) and \(\vP^*\) is equivalent to minimizing the KL divergence between the predicted, and ground truth ranking distributions, which yields the \rpllong defined in \Eqref{eqn:rlp}. \par

 The following Corollary presents an equivalence between the ListNet loss~\citep{cao2007learning} and the proposed \rpllong. 
\begin{corollary}
     \label{corr:rpl}
    (\textbf{\rplshort and the ListNet loss}) 
    Minimizing the \rpllong~is equivalent to optimizing for the ListNet top-1 probability loss (\Eqref{eqn:LNloss}~\citep{cao2007learning}) defined over modified scores  \(\sum_{k \in \sL_{j}} [\vf_{\vq_i}]_j\) and modified ground-truth scores \(\tilde{y}_{ij} = \sum_{k \in \sL_{j}} [\vy_{i}]_j\).
\end{corollary}
Intuitively, defining the modified scores in terms of the sum of all logits in \(\sL_j\), the set of indices of ground-truth scores lower than the logit at \(j\), ensures that the loss takes into account the interactions between the contextual embeddings contribution to the different logits. This is unique to the \alg setting, and Corollary~\ref{corr:rpl} shows that all guarantees derived for ListNet loss also hold for \rplshort. \par

We show in Section~\ref{sec:exps} that \alg trained with \rplshort yields significantly better ranking accuracy than existing pointwise and listwise loss functions, and leads to state-of-the-art performance on public and proprietary ranking benchmarks.

%% file: sections/exps.tex
\section{Experimental Validation}
\label{sec:exps}

\setlength{\tabcolsep}{4pt}
\begin{table*}[!t]
    \caption{Performance of \alg and the baseline methods on the SODQ and MS MARCO-Titles ranking datasets: All baselines and our method, \alg, use $\sim$60-100M parameter base models and are fine-tuned on the corresponding datasets, except for the large pre-trained models (indicated with asterisk (*)), which are used as is without any further fine-tuning on the two datasets. \alg surpasses similar-sized state-of-the-art methods fine-tuned for short-text ranking as well as large pre-trained models by at least 3\%.} 
    \label{table_acad_main}
      \centering
    \fontsize{9}{11}\selectfont 
        \begin{tabular}{@{}c|ccccccc}
        \toprule
 \multicolumn{3}{c}{\textbf{Method}} & \textbf{Parameters} & \multicolumn{2}{c}{\textbf{SODQ}} & \multicolumn{2}{c}{\textbf{MS MARCO-Titles}} \\
\midrule 
\multicolumn{3}{c}{} & {} & MAP@5 & MAP@10 & MRR@5 & MRR@10\\
\midrule
\textbf{Sparse Models} &  & BM25 & {--} & 32.80 & 39.26 & 23.71 & 24.57 \\
\midrule
\multirow{5}{*}{\makecell{\textbf{Encoder}}} & Early-interaction & monoBERT & {66M} & 46.79 & 48.04 & 30.89 & 32.47 \\
\cline{2-8}
& Late-interaction & ColBERT  & {109M} & {36.10} & {37.68} & {30.25} & 32.00 \\
\cline{2-8}
& \multirow{3}{*}{\makecell{Dual Encoders}} & DPR & {66M} & {47.32} & {48.48} & 28.78 & 30.87 \\
& & ANCE & {66M} & {48.31} & {49.41} & 28.48 & 30.53 \\
& & INSTRUCTOR* & 335M & 49.47 & 50.81 & 28.84 & 30.55 \\
\midrule
\multirow{2}{*}{\textbf{Encoder-Decoder}} & \multirow{2}{*}{Seq2Seq} & {RankT5-6L} & {74M} & {49.50} & {50.75} & {30.73} & {32.52}\\
& & {RankT5-base*} & {223M} & {45.66} & {49.47} & {27.87} & {29.75}\\
\midrule
\textbf{Decoder} & Ranking LLMs & {\makecell{RankingGPT-\\Llama2-7B*}} & {7B} & {47.64} & {50.62} & {28.66} & {30.47}\\
\midrule
\textbf{Ours} & \textbf{Joint Ranking} & \textbf{\alg-6L} & \textbf{66M} & \textbf{52.40} & \textbf{53.05} & \textbf{33.82} & \textbf{35.45} \\
        \bottomrule
    \end{tabular}
\end{table*}

\textbf{Datasets}:
We evaluate \alg on the publicly available on \textbf{Stack Overflow Duplicate Questions}~\citep{LinkSo} \textbf{(SODQ)} and a short-text version of
\textbf{MS MARCO}~\citep{dai2020context} datasets. 
While the SODQ dataset is used as is, for MS MARCO, a `short-text' variant of the query-webpage click dataset is constructed by exclusively considering webpage titles (subsequently referred to as \textbf{MS MARCO-Titles}).
This narrows the dataset's length distribution, aligning it more closely with the typical item lengths seen in sponsored search. 
We note that the standard metrics for MS MARCO passage ranking do not apply to MS MARCO-Titles, as they rely on the passage content as well as the title for ranking. 
Therefore, we report updated numbers for MS MARCO-Titles in Table~\ref{table_acad_main}, as appropriate. 
See Appendix~\ref{sup:dataset} for more details on the datasets.

\noindent \textbf{Evaluation Metrics}: To validate the efficacy of \alg for ranking tasks, we consider Mean Average Precision (MAP) and Mean Reciprocal Rank (MRR) as the evaluation metrics on both MS MARCO and SODQ datasets (see Appendix~\ref{sup:metrics} for metric definitions). Note that MAP is a more comprehensive version of MRR, which is specifically designed for scenarios where the test set contains only one positive item per test query. 
Hence, on MS MARCO, where there is only one relevant item per query, MAP@\(K\) and MRR@\(K\) are equivalent for any \(K\).

\noindent \textbf{Baselines}: We compare \alg{} with transformer based ranking approaches as well as sparse BoW methods such as \textbf{BM25}. 
Within transformer based ranking baselines, we consider state-of-the-art encoder based methods such as \textbf{monoBERT}~\citep{nogueira2020passage}.
These approaches employ \textit{early interaction} across query-item tokens (pointwise) in the form of self-attention and are widely used for ranking. We also compare \alg against more efficient encoders such as \textbf{ColBERT}~\citep{khattab2020colbert}.
\textit{Late-interaction} Models like ColBERT apply self-attention within query or item separately and combine the contextual embeddings across query and item later via a computationally cheap cross-interaction layer. We compare \alg with a ColBERT model consisting of 12 layers (109M parameters). Dual Encoders such as \textbf{ANCE}~\citep{Xiong20}, \textbf{DPR}~\citep{Karpukhin20}, and \textbf{INSTRUCTOR}~\citep{INSTRUCTOR} use contrastive learning to train the encoder and obtain dense representations of queries and items. Typically, a dot product between the query and item representations is used as a measure of relevance. INSTRUCTOR is a dual encoder approach, which obtains task-specific embeddings using instruction prompts to the model. Following ~\citet{INSTRUCTOR}, we use pre-trained and instruction-tuned INSTRUCTOR model for zero-shot evaluation on MS MARCO and SODQ.
We also compare against state-of-the-art sequence-to-sequence class of models containing encoder-decoder architecture, \textbf{RankT5}~\citep{rankt5}.
RankT5-base is a 24 layer model trained on MS MARCO passage ranking task and is taken as it is for evaluating as a baseline here.
For a fair comparison with 6 layer \alg model, we also fine-tune our own RankT5 model with 6 layers (encoder+decoder) on both MS MARCO-Titles and SODQ datasets.
For completeness, we also compare against a decoder only ranking model with Llama2-7B architecture, referred to as \textbf{RankingGPT-Llama2-7B}, available from ~\citet{zhang-etal-2024-two}. These models have been pre-trained as well as fine-tuned specifically for ranking tasks with huge amount of training data generated from larger and more accurate LLMs.\par
\noindent \textbf{Hyper-parameters}: \alg's tunable hyper-parameters include maximum item union length, $L_u$ and number of items per query, $N$. These hyper-parameter values are dictated by the application requirements and efficiency constraints. We use $N = 10$ and $L_u = 360$ on MS MARCO and $N = 30$ and $L_u = 242$ on SODQ. Hyperparameters used for the baselines are given in Appendix~\ref{sup:hyperparams}.

\subsection{Accuracy Results on Public Benchmarks}

Table~\ref{table_acad_main} shows the results of \alg and several competitive baselines on the SODQ and MS MARCO datasets. We observe that \alg achieves superior performance over the encoder based methods on both datasets, demonstrating the effectiveness of its listwise ranking approach. We also compare \alg with encoder-decoder based methods such as RankT5-6L, which encode queries and generate item rankings using a decoder.
Here, we note that RankT5-6L, which is fine-tuned with a listwise ranking loss, outperforms similar sized encoder based methods (monoBERT and ColBERT), highlighting the benefits of listwise modeling for ranking. 
However, we find that RankT5-base, a larger encoder-decoder model (24 layers, 223M parameters) that is fine-tuned on a long-text passage ranking task ~\citep{rankt5}, performs worse than smaller models that are specifically fine-tuned on short-text ranking. A similar observation is made for RankingGPT-llama2-7b, a large decoder only model (7B parameters) that is pre-trained and fine-tuned on large-scale ranking datasets ~\citep{zhang-etal-2024-two}, but performs similarly to much smaller models on short-text ranking. 
These results suggest that the Seq2Seq models are sensitive to the domain and length of the ranking tasks, and require careful fine-tuning and adaptation. To verify this, we fine-tune a RankT5-base model on the short-text ranking datasets and observe a significant improvement in performance (cf. Appendix~\ref{sup:hyperparams}). 
We also report that \alg, which uses a 6-layer BERT as the base encoder, has the same number of parameters as \sce, but is much faster and more accurate. Specifically, \alg can perform joint (listwise) inference for ranking over 4$\times$ faster than \sce, which requires multiple pointwise computations (cf. Section~\ref{SubSubSec_Efficiency}). Moreover, \alg performs up to 5\% more accurately than dual encoder methods and up to 20\% more accurately than sparse models like BM25. We perform more experiments on the comparison of pointwise and listwise loss functions in \alg in Section~\ref{sec:ablations}.

\subsection{Case Study on Sponsored Search Ads}
\label{case_study_ads}

In large-scale search and recommendation systems like sponsored search, the ranking model serves to weeding out bad retrievals and rank the prediction pool of different retrievers to select the top-k. We evaluate the effectiveness of \alg on this real-world task of matching user queries to relevant advertiser-bid keywords. 
A large scale dataset consisting of 1.8B query-keyword pairs was created by mining search engine logs (detailed in Appendix~\ref{sup:dataset}).

\noindent \textbf{Accuracy comparison}: 
As shown in Table~\ref{table_ads_main}\ (details on baselines in Appendix~\ref{App:AdsExp}), \alg{} improves over the in-production sparse neural model MEB~\citep{MicrosoftResearch2023} in MAP by over 13\%. Further, \alg also outperforms ANCE and TwinBERT by large margins in MAP, Precision, and Recall. We also assess \alg's ability to eliminate irrelevant items while preserving relevant ones. Table~\ref{table_ads_main} presents the negative and overall accuracy when retaining top 80\% of positive items per query.
\alg{} achieves 99.45\% negative accuracy, removing nearly all irrelevant items.\par

\noindent \textbf{Efficiency Gains}: 
We observe that \alg takes only 9.8 ms to score 700 keywords for a query on a A100 GPU. In contrast, \sce takes 41.3 ms for the same task, rendering it unsuitable for online deployment. This represents a \textit{more than 4-fold reduction in latency} compared to standard cross-encoder models. The latency gains are because \alg scores multiple items for a query in one shot by passing their concatenated tokens through the model. On the other hand, \sce scores each query-item pair independently necessitating 700 passes compared to \alg{'s} 7 passes. Additionally, CROSS-JEM provides 3$\times$ lower latency on GPUs than MEB on CPUs. This highlights CROSS-JEM's ability to leverage GPU acceleration for efficiency, unlike sparse models.

\alg{} achieves a high throughput of 17,200 query-keyword pairs per second. 
This is over 5$\times$ more than the 3,350 pairs per second for \sce. The massive throughput and latency gains show \alg{'s} ability to meet the computational demands of large-scale industrial systems without sacrificing accuracy.\par

\noindent \textbf{Online A/B testing}:
\alg{} was deployed in the ranking stage of a premier search engine to conduct A/B tests on live traffic. The ranking stage receives an average of 700 keywords and up to 1400 keywords in the 99th percentile, from a suite of retrieval algorithms. The control group consisted of a proprietary combination of late interaction, dense retrieval, and sparse-neural-network algorithms. \alg{} demonstrated a decrease in the quick-back-rate (users who close the ad quickly, indicating non-relevance) by over 1.8\%. Furthermore, as judged by expert judges, \alg{} improved the proportion of accurate predictions by 10.2\%.  

\setlength{\tabcolsep}{4pt}
\begin{table*}[!t]
    \caption{Comparison of \alg with production baselines on Sponsored Search Ads Dataset for ranking advertiser-bid keywords for a user query. \alg outperforms baseline methods (with a latency small enough to be deployed for real-time ranking) by 13\% in MAP. \alg filters $>$90\% irrelevant predictions at a threshold which retains 80\% of good predictions.}
    \label{table_ads_main}
      \centering
    \fontsize{9}{12}\selectfont
\begin{tabular}{@{}lccccccc}
\toprule
\textbf{Method} & \textbf{MAP@100} & \textbf{P@50} & \textbf{R@50}  & \textbf{AUC} & \textbf{Negative Accuracy} & \textbf{Overall Accuracy} \\
\midrule 
ANCE & 78.39 & 41.94 & 94.02 & 89.84 & 86.19 & 85.55 \\
TwinBERT & 83.56 & 43.50 & 95.36 & 92.10 & 90.60 & 88.58  \\
MEB & 84.38 & 42.94 & 94.65 & 91.77 & 84.59 & 84.40 \\
\textbf{\alg} & \textbf{97.48} & \textbf{45.76} & \textbf{99.07} & \textbf{99.41} & \textbf{99.45} & \textbf{95.27} \\
\bottomrule
    \end{tabular}
    \vskip1em
\end{table*}

\setlength{\tabcolsep}{4pt}
\begin{table}[!t]
    \caption{(a) Ablation on loss function in \alg; (b) Adding item token sequence information in \alg via positional encodings in the pooling layer.}
    \label{table_ablations}
      \centering
    \fontsize{9}{12}\selectfont
      \begin{tabular}{@{}c|cccc}
        \toprule
      \textbf{Method} & BCE & CE & ListNet & \rplshort \\
      \midrule 
      \textbf{MRR @10} & 31.46 & 32.03 & 30.27 & 35.45 \\
              \bottomrule 
              \multicolumn{5}{c}{(a)}
    \end{tabular}
    \vskip1em
    \begin{tabular}{@{}c|cc}
      \toprule
      \textbf{Method} &  \multirow{2}{*}{\makecell{Without Positional\\ Encodings}} & \multirow{2}{*}{\makecell{With Positional\\ Encodings}}\\
      & & \\
      \midrule 
      \textbf{MRR@10} & 35.45 & 35.68\\
      \textbf{MRR@5} & 33.82 & 33.98\\
      \bottomrule
      \multicolumn{3}{c}{(b)}
  \end{tabular}
\end{table}

\setlength{\tabcolsep}{4pt}
\begin{table}[!t]
    \caption{A comparison of latency between \alg and various baselines. The mean latency for scoring 700 items per query was computed on A100 GPUs for all models. We observe that \alg takes about $4\times$ lower inference time than standard cross-encoders (monoBERT).}
    \label{table_latency}
      \centering
    \fontsize{8}{12}\selectfont
      \begin{tabular}{@{}c|ccccc}
      \toprule
      \textbf{Method} & ANCE & ColBERT & monoBERT & RankT5-6L & \textbf{\alg} \\
      \midrule
      \(\substack{\text{\textbf{Latency}} \\ \text{\textbf{(ms)}}\downarrow}\) & 4.0 & 4.5 & 41.3 & 41.3 & 9.8 \\
      \bottomrule
  \end{tabular}
\end{table}

\subsection{Interpreting \alg{'s} Performance} \label{SubSubSec_Efficiency}

To better understand \textbf{the efficiency gains} in \alg, we analyze the effect of the significant token overlap amongst candidate items in the set \(\sK_{\vq}\) for a given query \(\vq\). We trained an ANCE~\citep{Xiong20} dense retriever on the same train set as above. For a query $\vq \in \sQ_{te}$, recall that $\sK_{\vq} = \{\vk_1, \vk_2, \ldots,\vk_{N}\}$ are the top $N$ (=100 in our experiments) items retrieved using ANCE. Let $\sT^L_{\vk_j}$ denote word-piece tokens in $\vk_j$ with a max-length of $L$. Let \(\sT_U\) denote the union of all tokens of items \(\vk_j \in \sK_{te}\). We compute the following  statistics:
\begin{align}
    m~\text{(mean total tokens)} & =  \frac{1}{|\sQ_{te}|} \sum_{\vq \in \sQ_{te}} \left(\sum_{j=1}^{N} |\sT^L_{\vk_j}|\right),~\text{and} \\
    N_u~\text{(mean size of the set \(\sT_U\))} & = \frac{1}{|\sQ_{te}|} \sum_{\vq \in \sQ_{te}} \left( \left| \bigcup_{j=1}^{N} \sT^L_{\vk_j}\right|\right).
\end{align}
Intuitively, \(m\) is the the sum of item token lengths on average, while \(N_u\) is the cardinaltity of the union set, averaged over the queries \(\vq\) for which the candidate items \(\sK_{\vq}\) were obtained. We hypothesize that, if the items \(\vk_i\) have significant overlap, \(m \gg N_u\). Statistically, the union size \(N_u\) is found to be at least 5$\times$ smaller than \(m\), indicating high redundancy, and correlates with observations on the \textit{sponsored search case study} (cf. Table~\ref{table_motivation} (b) in Appendix). \par

To further analyze this effect, 
we characterize the \textbf{time complexity} of \alg in terms of number of query tokens $L_{\vq} = \left|\sT_{\vq}\right|$, number of item tokens $L_{\vk} = \left|\sT_{\vk}\right|$, number of transformer layers $L$, number of candidate items $N$, and item union compression factor $C$ (approximated as $\frac{L_{\vk} \times N}{|\sT_{U_\vq}|}$). The time complexity for scoring all $N$ items jointly is
 $(L_{\vq} + L_{\vk}N/C)^2L$.
On the other hand, for the standard cross-encoder, the corresponding time complexity is
\(
    (L_{\vq}+L_{\vk})^2LN.
\)
In practice, the inference time depends on factors such as the implementation of the model, the hardware employed, and optimizations used (such as quantization). For the sponsored search setting considered in Section~\ref{sec:exps}, assuming $L_{\vq} \approx L_{\vk}$, letting $N = 100$, the maximum tokens in an item is 24 and maximum item union length used is 220, we have $C \approx (24*100)/220 \approx 11$. Then, the standard cross-encoder time complexity is \(400L_q^2L\), while that of \alg is \( 101.8L_q^2L\), which is 3.9$\times$ lower. These inference time gains are also reflected in the latency comparison of \alg against pointwise approaches such as monoBERT (cf. Table~\ref{table_latency}).


\subsection{Ablations}
\label{sec:ablations}

\textbf{Loss Function}: We demonstrate performance comparison of pointwise and listwise loss functions with \alg architecture in Table~\ref{table_ablations}. While listwise loss functions (Cross-Entropy and ListNet~\citep{cao2007learning}) either perform similar to or slightly better than pointwise losses (such as Binary Cross-Entropy), \alg trained with \rplshort performs much better than any pointwise or listwise loss.\\ 

\noindent \textbf{Incorporating Token Sequence Order Information}: \alg encodes the union of tokens from all ranking items, which enables fast inference by reducing the input sequence length. However, this also discards the original token order information within each item, treating them as bags of tokens. Though in our preliminary experiments, we observe that the loss of sequence information in the encoder has a negligible impact on the accuracy particularly in short-text ranking tasks (within 1\%, refer Appendix~\ref{sup:ablations_ads}).
However, token sequence information could be crucial in many ranking scenarios, and could get completely ignored in \alg. To address this limitation, we propose an extension of \alg that leverages sinusoidal positional encodings~\citep{vaswani2017attention} to inject the item sequence information back into the model via the selective pooling layer. The key idea is to preserve the original position indices of the tokens for each item in the ranking list, remove them during the encoding process, and then add them back to the corresponding token context vectors after the \alg encoder. The positional encodings are computed based on the original position indices and are summed with the token context vectors. The resulting vectors are then pooled together for classification as described in Section~\ref{sec:method}. This simple yet effective technique improves the MRR@10 by ~0.2\% on the MS MARCO dataset (cf. Table~\ref{table_ablations} (b)), without any additional latency. This technique could also be useful for improving \alg performance on long texts, which we leave for future work. 

%% file: sections/conclusion.tex
\section{Conclusion: Limitations and Future Work}
\label{sec:conclusion}

We introduced \alg, an accurate and efficient approach for joint ranking of a set of short-text items for a given query based on relevance. 
\alg effectively addresses the two major challenge in existing ranking architectures -- sub-optimal accuracy due to pointwise inference, and significantly higher computational cost. 
Our extensive evaluations on publicly available ranking benchmarks as well as large-scale sponsored search datasets reveal that \alg significantly outperforms the baselines, establishing a new state of the art. The scope of this work is primarily focused on the ranking of short texts, a common requirement in both industrial Sponsored Search applications and academic benchmarks, including tasks like matching queries with webpage titles and ranking duplicate questions. 
While the current work demonstrates significant gains on such short-text ranking tasks, the proposed approach could be adapted for long-text ranking by incorporating post-hoc positional encodings (cf.  Section~\ref{sec:ablations}) and more sophisticated attention mechanisms that account for longer document lengths. 
Exploring these adaptations is an area for future research.
\alg opens up new directions for designing accurate ranking architectures and algorithms, accounting for application-specific constraints.

%% file: sections/ethics.tex
\section{Ethical Considerations}

Our data usage and service provision practices have received approval from our legal and ethical boards. Socially, our research is significantly enhancing the efficiency and user experience for millions of people searching for goods and services online. This improvement is crucial in today’s context, as it enables contactless and time-efficient purchasing and delivery. Additionally, our work is boosting the revenue of numerous small and medium-sized businesses by expanding their market reach and lowering customer acquisition costs.

%% file: sections/suppl.tex
\onecolumn

\section{Proofs of Lemma~\ref{lemma:rpl} and Corollary~\ref{corr:rpl}} \label{app:proofs}

To derive the proof, we recall both  Lemma~\ref{lemma:rpl} and the \rpllong:
\begin{equation*}
    \cL^{\rplshort} = -\sum_{i=1}^{|\sQ_{tr}|}\sum_{j=1}^{N}\left(\sum_{k \in \sL_{j}}  [\vy_i]_k\right)\log\left(\mathrm{SoftMax}\left(\sum_{k \in \sL_{j}} [\vf_{\vq_i}]_{k} \right)\right),
\end{equation*}
where \(\tilde{s}_{ij} = \sum_{k \in \sL_{j}} [\vf_{\vq_i}]_j\) is the modified score associated with the query-item pair \((\vq_i,\vk_j)\). The following Lemma sheds light on the relationship between \rplshort and ranking probability distribution.
\begin{lemma*}
    (\textbf{\rpllong}) Let \(\vP \in \sR^{N\times N}\) denote a matrix with entries \(p_{jk}\) given by \(p_{jk}~=~\mathrm{Prob}\left(\text{ranking item \(\vk_j\) at location $k$}\right) \triangleq C \sum_{\ell \in \sL_{k}} [\vf_{\vq_i}]_{\ell}\), where \(C\) is a normalizing constant. Then, the \rpllong~maximizes the probability of ranking queries \(\sK_{\vq_i}\) in the ordering of the ground-truth ranking \(\vy_i\).
\end{lemma*}
\begin{proof}
Given $\vf_{\vq_i}$, the ground-truth scores \(\vy_i\), and a candidate item \(\vk_j\), we define the set $\sL_j = \{k \in \{1, N\}: [\vy_i]_k < [\vy_i]_j \}$, \textit{i.e.,} \(\sL_j\) comprises the indices \(k\) for which the ground truth score of \([\vy_i]_j\) is larger than the ground truth score at location \(k\). Additionally, without loss of generality, we assume that $\vf_{\vq_i}$ and \(\vy_i\) are normalized to have entries that lie in \([0,1]\). We have:
\begin{align*}
p_{jk}~=~\mathrm{Prob}\left(\text{ranking item \(\vk_j\) at location $k$}\right) \triangleq C \sum_{\ell \in \sL_{k}} [\vf_{\vq_i}]_{\ell},
\end{align*}
Where \(C\) is an appropriately chosed constant, such that \(\vP\) is a matrix with entries summing to one. Further, in the context of ranking with the logits/scores, we have:
\begin{align*}
p_{jk}&=\mathrm{Prob}\left(\text{ranking item \(\vk_j\) at location $k$}\right) \\
&= \mathrm{Prob}\left(\text{logit of item \(\vk_j\) > \{logits of all items $\vk_k$ such that $k\in \sL_k$ \}} \right).
\end{align*}
Since the entire analysis is presented in the context of a each query \(\vq_i\), for convenience, we abuse notation, and denote \([\vf_{\vq_i}]_j = [\vf]_j = f_j\). Then, we have
\begin{align}
p_{jk} &= \mathrm{Prob}\left(\text{logit of item \(\vk_j\) > \{logits of all items $\vk_k$ such that $k\in \sL_k$ \}} \right) \nonumber \\
&= \mathrm{Prob}\left(\text{\(f_j\) > \{\(f_k\) for all $k\in \sL_k$} \} \right) \nonumber \\ 
&= \mathrm{Prob}\left(f_j > \max_{k\in \sS_k}\{f_k\} \right) \nonumber  \\
& \approx \mathrm{Prob}\left(f_j > \sum_{\ell\in \sL_{k}}f_{\ell} \right), \label{eqn:App_ProbCond}
\end{align}
where the last step is a consequence of the norm-bound \(\|\vf\|_1 \leq \|\vf\|_{\infty} \). First, we note that, given the condition in \Eqref{eqn:App_ProbCond}, for all \(k\) ranked higher than \(j\), \(p_{jk} = 0\). Further, for all \(k\) ranked lower than \(j\), it suffices to show that the probability \(p_{jk}\) is non-zero, for \(j=k\). To verify this, consider two probabilities \(p_{i{k_1}}\) and \(p_{i{k_2}}\), with \(k_1\) ranked higher than \(k_2\). Let \(\bm{1}_{i{k_1}}\) denote the indicator of the event associated with \(p_{i{k_1}}\). We have
\begin{align*}
p_{i{k_2}} \approx \mathrm{Prob}\left(f_j > \sum_{\ell \in \sL_{k_2}}f_{\ell}~\Big\vert~ \bm{1}_{i{k_1}} \right) \mathrm{Prob}(\bm{1}_{i{k_1}}) = p_{i{k_1}}
\end{align*}
To build intuition for the this, assume without loss of generality that \(\vf\) have been sorted in a non-increasing manner. Let \(j = 1\) \( k_1 = 1\) and \(k_2 = 3\). Then, it is clear to see that 
\begin{align*}
&\mathrm{Prob}\left(f_1 > f_4 + f_5 + \ldots ~\Big\vert~ f_1 > f_2 + f_3 + f_4 + f_5 + \ldots \right)\mathrm{Prob}\left(f_1 > f_2 + f_3 + \ldots \right) \\
&= \mathrm{Prob}\left(f_1 > f_2 + f_3 + \ldots \right).
\end{align*}
Given \(\vP\) built as described above, and \(\vP^*\) defined similarly over scores \(\vy\), from the discussion above, we see that minimizing the distance between \(\vP\) and \(\vP^*\) is equivalent to minimizing the distance between the vectors \(\vp = \mathrm{Diag}\left(\vP\right)\) and \( \vp^* = \mathrm{Diag}\left(\vP^*\right)\). We observe that, when normalized, the entries of \(\vp\) correspond to the probability of ranking item \(\vk_j\) and rank \(j\), given sorted vectors \(\tilde{\vf}\) and \(\tilde{\vy}\). Then, the \rplshort is the binary cross entropy loss defined between \(\vp\) and \(\vp^*\), and represents minimizing the KL-divergence between the predicted ranking probability distribution \(\tilde{\vf}\) and the ground-truth ranking distribution \(\tilde{\vy}\), up to a normalizing constant factor. This completes the proof of Lemma~\ref{lemma:rpl}.
\end{proof}
\textbf{Proof of Corollary~\ref{corr:rpl}}: To link the \rpllong to the ListNet loss~\citep{cao2007learning}, we recall that:
\begin{align*}
    \cL^{LN} = - \sum_{i=1}^{|\sQ_{tr}|}\sum_{j=1}^{N} P_{\vy,j} \log \left(P_{\vs,j}\right),
\end{align*}
where \(\displaystyle P_{\vx,j} = \frac{\Phi\left([\vx]_j\right)}{\sum_{\ell=1}^{N}\Phi\left([\vx]_{\ell}\right)}\)
Setting \([\tilde{\vs}_{\vq_i}]_j = \tilde{s}_{ij} = \sum_{{\ell} \in \sL_{j}} [\vf_{\vq_i}]_{\ell}\),~and~\([\tilde{\vy}_{i}]_j = \tilde{y}_{ij} = \sum_{{\ell} \in \sL_{j}} [\vy_{i}]_{\ell}\) and a \(\Phi\) that result in mapping \(\tilde{\vs}\) and \(\tilde{\vy}\) to be valid probabilities distributions, we see the equivalence between the ListNet loss, and the proposed \rpllong. 

\section{The \alg Inference} \label{sup:algo} 
The procedure for forward pass of a \alg model during training as well as inference is outlined in Algorithm \ref{algo-forward-pass-crossjem}, while the procedure to obtain the attention map over the item union set \(\sT\) is described in Algorithm~\ref{algo-get-kumaps}.

\input{tables.tex/pseudo-code}

\section{Datasets}
\label{sup:dataset}
\textbf{MS MARCO}: MS MARCO document re-ranking dataset contains queries and their clicked web pages consisting of webpage title, URL, and passage. We create a \'short-text\' version of the dataset by considering only the titles of the webpages, making the length statistics of the dataset better aligned with the real-world applications of ranking in sponsored search. We experiment with HDCT~\citep{dai2020context} retriever based training dataset consisting of 0.37M training queries and top 10 predictions from HDCT along with their ground truth click labels. We use the dev set to report our metrics as this set has ground truth labels available for evaluation. We use $\sim$3.7M training pairs available in MS MARCO HDCT dataset sourced from~\citep{dai2020context} as described above for training \alg and baselines using DistilBERT~\citep{sanh2019distilbert}. The target scores for all training pairs are obtained from a \sce model trained on binary ground truth click data with BERT-Base~\citep{bert-base} as the base encoder for MS MARCO dataset. 

\textbf{SODQ}: Stack Overflow Duplicate Questions dataset involves ranking questions on Stack Overflow as duplicates or not with the tags Java, JavaScript and Python ~\citep{LinkSo}. It is also one of the re-ranking datasets on the popular MTEB benchmark~\citep{muennighoff2022mteb}. StackOverflowDupQuestions on MTEB benchmark is the only shoft text re-ranking dataset with both training and evaluation data available, and is hence used in our experiments. Similar to our experiments on MS MARCO dataset, the target scores for training \alg as well as baselines are obtained from a BERT~\citep{bert-base} based \sce model trained on binary relevance of duplicate questions. 

\textbf{Sponsored Search Dataset}: The training dataset for sponsored search query to advertiser matching task is created using a BERT-Large based \sce model trained on manually labeled and good-click data as the teacher model. 
A query-item (advertiser bid keyword) pair in the good click data is obtained when the user clicked on the ad corresponding to an advertiser keyword in response to their query, and did not close the ad quickly indicating they found it relevant. 
This BERT-Large teacher model was used to score 100 predicted items each for 18.6M queries on the search engine during a time period. This resulted in around 1.8B query-item pairs with scores in 0 to 1 range as training data for \alg{} and all baselines in Table~\ref{table_ads_main}. 

\section{Metrics}
\label{sup:metrics}
\begin{itemize}[leftmargin=*]
    \item \textbf{Mean Average Precision (MAP)}: This is a ranking metric defined as the mean of Average Precision (AP) over the positive and negative detected classes: 
    \begin{equation}
    \frac{1}{|Q|} \sum_{u=1}^{|U|} \left(\frac{1}{m} \sum_{k=1}^{N} P_u(k) \cdot rel_u(k)\right)
    \end{equation}
    where $|Q|$ is the total number of queries, $P_u(k)$ is the precision at cut-off $k$ in the list, $rel_u(k)$ is an indicator function equaling 1 if the item at rank $k$ is a relevant document, otherwise zero.

    \item \textbf{Mean Reciprocal Rank (MRR)}: Rank is defined as the position of the first relevant item in the ranked list. MRR is hence defined as below:
    \begin{equation}
        \frac{1}{|Q|} \sum_{i=1}^{|Q|} \left( \frac{1}{rank_i}\right)
    \end{equation}

    \item \textbf{Accuracy}: Positive, Negative, and Overall Accuracy denote the proportion of positive, negative, and overall instances, respectively, in the test set that are accurately identified.
    \item \textbf{Area Under the ROC Curve (AUC-ROC)}: The ROC curve is a plot of True Positive Rate (TPR) or sensitivity against False Positive Rate (FPR) at different thresholds.
\end{itemize}

\section{Baselines and Hyperparameters}
\label{sup:hyperparams}
For baselines monoBERT, DPR, and ANCE, we tune the following hyperparameters based on the validation set accuracy: learning rate, weight decay, number of epochs. ColBERT is trained and evaluated with default hyperparameters provided by the authors\footnote{https://github.com/stanford-futuredata/ColBERT}. We use the pre-trained checkpoint\footnote{https://huggingface.co/hkunlp/instructor-base} and code-base\footnote{https://github.com/xlang-ai/instructor-embedding} provided by the authors for INSTRUCTOR model, and test its zero-shot performance using the instructions mentioned in the paper. \alg is trained with exactly same setting as \sce: learning rate of 1e-4, linear learning rate scheduler, and AdamW optimizer. Hyperparameters specific to \alg ($N$ and $L_u$) are provided in Table~\ref{table_hyperparams}.
The metrics for BM25 on SODQ are taken from~\citep{LinkSo}, while they are computed following ~\citet{bm25} on MS MARCO.

\textbf{Fine-tuning RankT5-base (24L) for Short-text Ranking}:
We fine-tune the model checkpoint available from ~\citet{zhang-etal-2024-two} on short-text ranking benchmarks and note the performance improvements in Table~\ref{table_rankt5}.

\setlength{\tabcolsep}{4pt}
\begin{table*}[!t]
    \caption{Improvement in performance of a large Seq2Seq model, RankT5-base after fine-tuning on short-text ranking benchmarks} 
    \label{table_rankt5}
      \centering
    \fontsize{9}{12}\selectfont 
        \begin{tabular}{@{}c|cccc}
        \toprule
    {\textbf{Method}} & \multicolumn{2}{c}{\textbf{SODQ}} & \multicolumn{2}{c}{\textbf{MS MARCO}} \\
\midrule 
{} & MAP@5 & MAP@10 & MRR@5 & MRR@10\\
\midrule
\textbf{RankT5-base (pre-trained)} & 45.66 & 49.47 & 27.87 & 29.75 \\
\textbf{RankT5-base (fine-tuned)} & 55.9 & 56.8 & 33.72 & 35.14 \\
        \bottomrule
    \end{tabular}
\end{table*}

\subsection{Compute}
\label{sup:compute}
All baselines on MS MARCO and SODQ datasets including \alg were trained on 8 V100 GPUs. Experiments on proprietary Sponsored Search dataset were conducted on larger GPU cluster with 16 V100s.

\setlength{\tabcolsep}{4pt}
\begin{table}[!t]
    \caption{Hyperparameters used in \alg.}
    \label{table_hyperparams}
      \centering
    \fontsize{9}{12}\selectfont
\begin{tabular}{@{}l|c|c|c}
\toprule
\textbf{Hyperparam} & \textbf{SODQ} & \textbf{MS MARCO} & \textbf{Sponsored Search} \\
\midrule 
$N$ & {30} & {10} & {100} \\
$L_u$ & {265} & {360} & {262} \\
\bottomrule
    \end{tabular}
\end{table}

\section{Experiments: Sponsored Search Dataset}\label{App:AdsExp}
We compare \alg against methods that can be deployed for real-time ranking including ANCE, MEB, and TwinBERT. 
TwinBERT is a lighter version of ColBERT. 
It applies an MLP layer to individual query and keyword embeddings, unlike ColBERT which considers interactions along all token embeddings. This makes TwinBERT more efficient and practical for real-world systems due to lower storage requirements.
Table~\ref{table_ads_main} shows the comparison of \alg with baseline methods in production where \alg outperforms existing methods by large margins.
\label{sup:ablations_ads}

\textbf{Ablation on Number of items per Query ($N$)}: From Table \ref{table_abl_k}, we vary the number of items scored per query from 10 to 100. With more items, the token overlap increases, providing \alg more opportunity for joint modeling. Correspondingly, we observe gains in negative accuracy and AUC as items per query increase.

\textbf{Ablation on Encoder $\cE_{\vtheta}$ in \alg}: 
Table \ref{table_abl_trans} shows that even with a smaller encoder, CROSS-JEM provides significant accuracy gains over sparse models like MEB while having low latency.

\textbf{Ablation on Sequence Information}: We analyze the effect of the ordering/sequencing of the item text on classification accuracy using a cross-encoder model; and a train dataset consists of about 100M query-item pairs \((\vq,\vk)\), drawn from $\sQ_{tr} \times \sI_{tr}$, mined from proprietary search engine logs. Given \((\vq,\vk)\), we compare two cross-encoder \(\gE\) models: 1) $\gE_{CE}$: Standard cross-encoder scoring the pairs $(\vq, \vk)$, and 2) $\gE^{\prime}_{CE}$: Cross-encoder trained to score pairs $(\vq, \vk^{\prime})$, where the item $\vk^{\prime}$ is obtained by sorting the tokens in $\vk$ alphabetically.
  
The hypothesis is that, if the cross-encoders $\gE_{CE}$ and $\gE_{CE}^{\prime}$ have similar scoring accuracy, then the sequence ordering is relatively less informative for this task. While testing $\gE_{CE}^{\prime}$ is evaluated on \((\vq,\vk^{\prime})\) pairs \(\vk^{\prime}\) is drawn from $\sI_{te}$, with its tokens sorted alphabetically. Table~\ref{table_motivation}(a) shows how the variants perform on a test set of 10M pairs. We observe that, when the sequence information in \(\vk\) is discarded, both the mean average precision (MAP) and accuracy are within 1\% of the case when the sequence information is retained.

\input{tables.tex/motivation_tables}

\setlength{\tabcolsep}{4pt}
\begin{table}[!t]
    \caption{Variation in accuracy on varying the number of items scored per query by \alg. We observe only minor variation in changing the number of items to be ranked at inference time. This observation is useful in real-world ranking which receive item candidates from a set of retrieval algorithms and hence the number of items to be scored can vary with the query.}
    \label{table_abl_k}
      \centering
    \fontsize{9}{12}\selectfont
\begin{tabular}{@{}l|ccc}
\toprule
\textbf{$N$} & \textbf{Negative Accuracy} & \textbf{Overall Accuracy} & \textbf{AUC} \\
\midrule 
10 & 99.18 & 94.94 & 99.24 \\
20 & 99.37 & 94.96 & 99.34 \\
50 & 99.52 & 94.82 & 99.40 \\
80 & 99.56 & 94.71 & 99.51 \\
100 & 99.45 & 95.27 & 99.42 \\
\bottomrule
    \end{tabular}
\end{table}

\setlength{\tabcolsep}{4pt}
\begin{table}[!t]
    \caption{Variation in accuracy with base encoder $\cE_{\vtheta}$ in \alg. 
    }
    \label{table_abl_trans}
      \centering
      \fontsize{9}{12}\selectfont
\begin{tabular}{@{}l|cccc}
\toprule
\textbf{Encoder} & \textbf{Negative Accuracy} & \textbf{Overall Accuracy} & \textbf{AUC} & \textbf{Latency CPU} \\
\midrule 
TinyBERT - 2 layer & 96.52 & 92.76 & 96.67 & 114.3 \\
DistilBERT - 6 layer & 99.45 & 95.27 & 99.61 & 744.7 \\
\bottomrule
    \end{tabular}
\end{table}
\newpage
\section{Qualitative Analysis}
\setlength{\tabcolsep}{4pt}
\begin{table*}[!htb]
    \caption{A comparison of the top-5 ranked items obtained using \alg{'s} listwise modeling and a pointwise ranking model ~\citep{nogueira2020passage}. Relatively more \textit{generic} (but still relevant) items are ranked higher in baseline predictions, owing to their frequency in training data, token-level matching and other biases, which are circumvented when using listwise architectures capable of evaluating all the shortlisted items in a single forward pass, and rank the more relevant (and specific) items higher.}
    \label{table_AblCompares}
      \centering
    \fontsize{6.5}{10}\selectfont
    \begin{tabular}{P{.245\linewidth}|P{.38\linewidth}P{.31\linewidth}}
      \toprule

      \textbf{Query} & \textbf{Top-5 Ranked Item in the Proposed Approach (\alg)}  & \textbf{Top-5 Ranked Item in a Baseline Cross Encoder}   \\
      \midrule
      \multirow{5}{*}{different foods of oaxaca mexico} & \good{`the foods of oaxaca'} & \meh{`culinary tales: the kinds of food mexicans eat'} \\
      & \good{`oaxacan cuisine'} & \meh{`mexican christmas foods'} \\
      & \good{`exploring oaxacan food'} & \meh{`mexican cuisine'} \\
      & \good{`authentic recipes from oaxaca, mexico'} & \meh{`culture: food and eating customes in mexico'} \\
      & \meh{`6 things you'll love about oaxaca'} & \meh{`popular food in mexico'}  \\
      \midrule  
      \multirow{5}{*}{what is the bovine growth hormone} &  \good{`recombinant bovine growth hormone'} & \meh{`growth hormone'} \\
      &  \good{`what is rbgh?'} & \bad{`human growth hormone'} \\
      & \good{`rbgh'} & \bad{`alternative names for growth hormone'} \\
      & \meh{`bovine growth hormone and milk : what you need to know'} & \bad{`human growth hormone and insulin are friends'} \\
      & \good{`what is rbst?'} & \bad{`growth hormone (somatotropin)'} \\
      \midrule  
      \multirow{5}{*}{what is the state nickname of new mexico} & \good{`what are the nicknames of the state of new mexico?'} & meh{`the state of new mexico'} \\
      & \good{`state nicknames new mexico - south carolina and their explanation'} & meh{`state of new mexico'} \\
      & \good{`what is new mexico's nickname?'} & \meh{`new mexico'} \\
      & \good{`new mexico state names (etymology of names)'} & \bad{`new mexico state university'} \\
      &  \meh{`the state of new mexico'} & \bad{`state of mexico'} \\
      \midrule   
      \multirow{5}{*}{is the elliptical bad for your knees} & \good{`does an elliptical make bad knees worse?'} & \meh{`what exercises can help relieve knee pain?'} \\
      & \good{`are the elliptical machines bad for your knees?'} & \bad{`4 bad exercises for bad knees'} \\
      & \good{`is an elliptical the best machine for knees that are chronically painful?'} & \bad{`how to treat a knee sprain'} \\
      & \good{`why does my knee hurt on an elliptical machine?'} & \bad{`how to strengthen legs with bad knees'} \\
      & \good{`elliptical machine is good for osteoarthritis of the knee!'} & \bad{`yoga bad for your knees, indian doctor warns'} \\
      \midrule 
      \multirow{5}{*}{what do you use for oxygen facial machines} & \good{`oxygen facials and other skincare services'} & \bad{`using oxygen safely'} \\
      & \meh{`oxygenating facial treatment'} & \bad{`the oxygen machine'} \\
      & \good{`oxygen facial : home kits, beauty benefits, side effects, process'} & \bad{`shop cpap and oxygen'} \\
      & \meh{`benefits of oxygen facial'} & \bad{`oxygen concentrators and generators'} \\
      & \meh{`4 beauty - boosting benefits of oxygen facials'} & \bad{`oxygen concentrators \& generators'} \\
      \bottomrule
  \end{tabular}
\end{table*}

%% file: tables.tex/pseudo-code.tex

\begin{algorithm}
\caption{Method to create attention masks for item $k_{i}$. \textbf{Input}: $T_{q}$: Tokenized Query, $T_{U}$: Tokens in item Union Set, $k$: tokenized $k_{i}$ keyword. \textbf{Output}: AttMask: Attention Mask for the keyword}
\label{algo-get-kumaps}
\begin{algorithmic}[1]
\Procedure{GetKUAttentionMask}{$T_{q},T_{U},k$}  
    \State AttMask $\gets []$
    \For{$i$ from $0$ to LEN($T_{q}$) - 1}
        \State AttMask.\Call{AddItem}{1}
    \EndFor
    \For{$i$ from $0$ to LEN($T_{U}$) - 1}
        \If{$T_{U}[i]$ in $k$}  
            \State AttMask.\Call{AddItem}{1}
        \Else
            \State AttMask.\Call{AddItem}{0}
        \EndIf
    \EndFor
    \State \Return AttMask
\EndProcedure
\end{algorithmic}
\end{algorithm}

\begin{algorithm}
\caption{Getting relevance scores for a query and retrieved set of items using CROSS-JEM. \textbf{Input}: Query $q$, retrieved set of $N$ items for $q$ where $K_{q} = \{k_0, k_1, ..., k_{N-1}\}$. \textbf{Output}: Scores $S= \{s_0, s_1, ..., s_{N-1}\}$}
\label{algo-forward-pass-crossjem}
\begin{algorithmic}[1]

\State  $T_q \gets \Call{Tokenize}{Q}$ \Comment{Tokenize the query}
\State  $T_U \gets \{\}$ 
\State  $K_{tokens} \gets \left[\right]$ \Comment{Store tokenized items}

\For{$k$ in $K_{q}$}
    \State $k_{tokens} \gets \Call{Tokenize}{k}$
    \State $T_{U} \gets \Call{Union}{T_U, k_{tokens}}$
    \State  $K_{tokens} . \Call{AddItem}{k_{tokens}}$
\EndFor

\State $T_{U} \gets \Call{Sorted}{T_U}$
\State KUAttMask $ \gets \left[\right]$ \Comment{KUAttMask: item Union Attention Mask}

\For{$i$ from $0$ to $N-1$}
    \State AttMask $\gets \Call{GetKUAttentionMask}{T_q, T_U, K_{tokens}[i]} $ (cf. Algorithm~\ref{algo-get-kumaps})
    \State KUAttMask.\Call{AddItem}{AttMask}
\EndFor

\State sepToken $\gets \Call{Tokenize}{\left[SEP\right]}$ \Comment{Token id for [SEP] token}
\State encInpToks $\gets T_Q $ \Comment{Tokens to be passed through the Encoder}
\State encInpToks.\Call{AddItem}{sepToken} 

\For{$t_U$ in $T_{U}$}
    \State encInpToks.\Call{AddItem}{$t_U$}
\EndFor

\State $E \gets$ \Call{Encoder}{encInpToks}  

\State $S \gets$ \Call{SelectivePooling}{$E$, KUAttMask} \Comment{$S \in \mathbb{R}^{N \times d}$}

\State $S  \gets$ \Call{Classifier}{$S$}  \Comment{$S \in \mathbb{R}^N$}

\State \Return $S$
\end{algorithmic}
\end{algorithm}

%% file: tables.tex/motivation_tables.tex


\setlength{\tabcolsep}{4pt}
\begin{table}[!t]
    \caption{Sponsored search dataset statistics motivating \alg architecture design. (a) Cross-encoder trained with ordered tokens $(\gE_{CE})$ and alphabetically sorted tokens for all items $(\gE^{\prime}_{CE})$. The small performance delta (between rows) indicates that sequence ordering is not critical for short-text. (b) The mean total tokens \(m\) in a retrieved item set \(Q_{t_e}\) and the mean size \(N_u\) of tokens in the union of \(Q_{t_e}\). The ratio of \(m\) to \(N_u\) being \(\sim5\) indicates strong token overlap.}
    \label{table_motivation}
      \centering
    \fontsize{9}{12}\selectfont
\begin{tabular}{@{}c|c}
\toprule
\textbf{Algorithm} & \textbf{MAP@100} \\[6pt]
\midrule
$\gE_{CE}$ & 93.03 \\[6pt]
$\gE^{\prime}_{CE}$ & 92.76 \\
\bottomrule
\multicolumn{2}{c}{(a)}
    \end{tabular}
    \hspace{0.5cm}
\begin{tabular}{@{}c|c|c}
\toprule
\textbf{Max-length in word-piece ($L$)} & {\(\mathbf{m}\)} & \textbf{\(\mathbf{N_u}\)}\\
\midrule
12 & 498.38 & 93.70 \\
16 & 499.34  & 94.02 \\
32 & 501.52  & 94.60 \\
\bottomrule
\multicolumn{3}{c}{(b)}
\end{tabular}
\end{table}